\documentclass[journal,twoside,web]{_Aux/ieeecolor}
\usepackage{_Aux/generic}
\usepackage{times}
\usepackage{setspace}
\usepackage{url}
\spacing{1}
\usepackage[utf8]{inputenc}
\usepackage[T1]{fontenc}
\usepackage{graphicx}		
\usepackage{wrapfig}
\usepackage[format=plain,font=footnotesize,labelfont=bf,labelsep=period]{caption}
\usepackage{sidecap} 
\usepackage{subfig}
\usepackage[export]{adjustbox}
\usepackage[font=small]{caption}
\usepackage{float}

\usepackage{amsmath} 
\usepackage{amssymb}  
\usepackage{amsthm}
\usepackage{mathtools}
\usepackage[normalem]{ulem}
\usepackage{paralist}	
\usepackage[space]{grffile} 
\usepackage{color}
\let\labelindent\relax
\usepackage{enumitem}
\usepackage{bm}
\usepackage{cancel}

\newtheorem{theorem}{Theorem}

\theoremstyle{definition}
\newtheorem{definition}{Definition}
\theoremstyle{remark}
\newtheorem{remark}{Remark}
\theoremstyle{definition}

\theoremstyle{definition}

\newtheorem{example}{Example}

\newcommand{\R}{\mathbb{R}}
\newcommand{\C}{\mathcal{C}}

\newcommand{\K}{\mathcal{K}}

\definecolor{blue}{RGB}{38,38,134}
\definecolor{darkblue}{RGB}{0,0,102}
\definecolor{lightblue}{RGB}{77,77,148}

\definecolor{gold}{RGB}{234, 170, 0}
\definecolor{metallic_gold}{RGB}{139, 111, 78}

\renewcommand{\cal}[1]{\mathcal{ #1 }}
\newcommand{\mb}[1]{\mathbf{ #1 }}

\newcommand{\der}[2]{\frac{\mathrm{d} #1 }{\mathrm{d} #2 }}
\newcommand{\derp}[2]{\frac{\partial #1 }{\partial #2 }}

\DeclareMathOperator*{\argmin}{argmin}

\usepackage{xfrac}

\usepackage{dblfloatfix}

\begin{document}

\title{\LARGE \bf
Safe Controller Synthesis with Tunable \\ Input-to-State Safe Control Barrier Functions}

\author{Anil Alan$^{1}$, Andrew J. Taylor$^{2}$, Chaozhe R. He$^{1,3}$, G\'abor Orosz$^{1,4}$, and Aaron D. Ames$^{2}$
\thanks{This research is supported in part by the National Science Foundation, 
CPS Award \#1932091.}
\thanks{$^{1}$A. Alan, C. R. He, and G. Orosz are with the Department of Mechanical Engineering, University of Michigan, Ann Arbor, MI 48109, USA ${\tt\small \{anilalan,hchaozhe, orosz\}@umich.edu}$}%
\thanks{$^{2}$A. J. Taylor and A. D. Ames are with the Department of Mechanical and Civil Engineering, California Institute of Technology, Pasadena, CA 91125, USA ${\tt\small \{ajtaylor, ames\}@caltech.edu}$}%
\thanks{$^{3}$C. R. He is also with Navistar, Inc. Lisle, IL 60532, USA 
${\tt\small Chaozhe.He@navistar.com}$  }
\thanks{$^{4}$G. Orosz is also with the Department of Civil and Environmental Engineering, University of Michigan, Ann Arbor, MI 48109, USA}%
}

\maketitle

\thispagestyle{empty}         

\begin{abstract}
To bring complex systems into real world environments in a safe manner, they will have to be robust to uncertainties---both in the environment and the system.  This paper investigates the safety of control systems under input disturbances, wherein the disturbances can capture uncertainties in the system.  Safety, framed as forward invariance of sets in the state space, is ensured with the framework of control barrier functions (CBFs). Concretely, the definition of input-to-state safety (ISSf) is generalized to allow the synthesis of non-conservative, tunable controllers that are provably safe under varying disturbances. This is achieved by formulating the concept of tunable input-to-state safe control barrier functions (TISSf-CBFs), which guarantee safety for disturbances that vary with state and, therefore, provide less conservative means of accommodating uncertainty. The theoretical results are demonstrated with a simple control system with input disturbance and also applied to design a safe connected cruise controller for a heavy duty truck.
\end{abstract}


\section{Introduction} \label{sec:intro}

\IEEEPARstart{S}{afety} is of the utmost importance for control systems, often prioritized over other performance requirements. A formal definition of safety has been proposed via the forward invariance of sets in the state space. Forward invariance can be ensured using \emph{barrier certificates} \cite{prajna2004safety} and \emph{barrier functions} \cite{ames2014control,Ames:2015IFAC}.
The extension of the latter to \emph{control barrier functions} (CBF) provides a tool for control design by imposing an easy-to-compute condition over a desired safe set. A recent survey on CBFs can be found in \cite{Ames:2019ECC}, and alternative methods for safety-critical control in \cite{Chen:18cdc,Leung:20}.

Among other relevant applications such as multi-agent systems \cite{GloCorEge:2017} and robotics \cite{AyuKou:17}, automated vehicles stand out as a natural candidate for safety-critical control. Due to recent developments of optical sensors and vehicle-to-everything (V2X) communication modules, many safety hazards in traffic can be detected. Thus, the goal of control design is to prevent safety breaches while utilizing sensory and V2X information. Examples of the use of control barrier functions include adaptive and connected cruise control \cite{ames2014control,Chaozhe:18} and lane keeping \cite{Xuetal:2020} problems. The effectiveness of the safety-critical control is typically demonstrated using simulations that may be transferred to the real world assuming that the systems model is accurate.

Uncertainties such as unmodeled dynamics and unknown input disturbances pose risks to guaranteeing safety in the real-world implementations. Robust CBF methods have been proposed to address this problem \cite{Emam19cdc,Nguyen:20arxiv,jankovic2018robust}. 
We focus on the concept of \emph{input-to-state safety} (ISSf) first introduced in \cite{Romdlony:16} and extended in \cite{Kolathaya:2019ISSf} to address bounded disturbances in the system's input. In this setting safety in the presence of disturbances is redefined as the forward invariance of a larger set. While control design under an unknown bounded input disturbance is possible utilizing \emph{input-to-state safety control barrier functions} (ISSf-CBF),
this approach lacks flexibility in design and often yields conservative results.

In this paper we revisit the fundamental definition of ISSf and ISSf-CBF and generalize them to enable a tunable control design. Our main results introduces \emph{tunable input-to-state safety control barrier function} (TISSf-CBF), a generalized version of ISSf-CBF, that permits controllers to provide safety guarantees in the presence of bounded disturbances in the input while reducing conservatism. In particular, it allows one to  tune the size of the larger invariant set so that it approximates the safe set of the undisturbed system without significantly impacting performance. Furthermore, our approach may be combined with existing methods using robust CBFs \cite{jankovic2018robust} when disturbances may be decoupled into external disturbances and disturbances in the system input. The applicability of our proposed approach is demonstrated using a simple example as well as the real-world application of a connected cruise controller for a heavy duty vehicle.

\section{Background and Motivation} \label{sec:background}

This section presents a review of safety and control barrier functions, followed by the notion of input-to-state safety in the presence of input disturbances. These theoretical concepts are illustrated with a simple example.  

\subsection{Safety and Control Barrier Functions}

We consider a nonlinear control-affine system: 
\begin{equation}
    \label{eq:openloop1}
    \dot{\mb{x}} = \mb{f}(\mb{x})+\mb{g}(\mb{x})\mb{u},
\end{equation}
with state ${\mb{x}\in\R^n}$, input ${\mb{u}\in\R^m}$, and functions ${\mb{f}:\R^n\to\R^n}$ and ${\mb{g}:\R^n\to\R^{n\times m}}$ assumed to be locally Lipschitz continuous on $\R^n$. Using a locally Lipschitz continuous state feedback controller ${\mb{k}:\R^n\to\R^m}$, with ${\mb{u}=\mb{k}(\mb{x})}$, yields the closed loop system:
\begin{equation}
    \label{eq:closedloop1}
    \dot{\mb{x}} = \mb{f}(\mb{x})+\mb{g}(\mb{x})\mb{k}(\mb{x}).
\end{equation}
As the functions $\mb{f}$, $\mb{g}$, and $\mb{k}$ are locally Lipschitz continuous, for any initial condition ${\mb{x}_0 \triangleq \mb{x}(0) \in \R^n}$, there exists a time interval ${I(\mb{x}_0)=[0,t_{\rm max})}$ such that $\mb{x}(t)$ is the unique solution to \eqref{eq:closedloop1} on $I(\mb{x}_0)$; see \cite{perko2013differential}.

We define the notion of safety in this context as forward invariance of a set in the state space. Specifically, suppose there exists a set ${\C\subset \R^n}$ defined as the 0-superlevel set of a continuously differentiable function ${h:\R^n \to \R}$:
\begin{align}
    \label{eq:C1} \C &\triangleq \left\{\mb{x} \in \R^n : h(\mb{x}) \geq 0\right\}, \\
    \label{eq:C2} \partial\C &\triangleq \{\mb{x} \in \R^n : h(\mb{x}) = 0\},\\
    \label{eq:C3} \textrm{Int}(\C) &\triangleq \{\mb{x} \in \R^n : h(\mb{x}) > 0\}.
\end{align}
The set $\C$ is said to be \emph{forward invariant} if for any initial condition ${\mb{x}_0 \in \C}$, ${\mb{x}(t)\in\C}$ for all ${t\in I(\mb{x}_0)}$. In this case, we call the system \eqref{eq:closedloop1}  \emph{safe} with respect to the set $\C$, and refer to $\C$ as the \emph{safe set}.

A continuous function ${\alpha:[0,\infty)\to[0,\infty)}$ is said to be \emph{class $\cal{K}_\infty$} (${\alpha\in\cal{K}_{\infty}}$) if $\alpha$ is strictly monotonically increasing with ${\alpha(0)=0}$ and ${\lim_{r\to\infty}\alpha(r)=\infty}$, and a continuous function ${\alpha:\R\to\R}$ is said to be \emph{extended class $\cal{K}_\infty$} (${\alpha\in\cal{K}_{\infty,\rm e}}$) if it belongs to $\cal{K}_\infty$ and ${\lim_{r\to-\infty}\alpha(r)=-\infty}$. With these definitions, control barrier functions, as defined in \cite{ames2017control}, provide a tool for synthesizing controllers that enforce the safety of $\C$ (where a strict inequality is used for the reasons outlined in \cite{jankovic2018robust}).
\begin{definition}[\textit{Control Barrier Function (CBF)} \cite{ames2017control}]
Let ${\C\subset\R^n}$ be the 0-superlevel set of a continuously differentiable function ${h:\R^n\to\R}$ with ${\derp{h}{\mb{x}}(\mb{x}) \neq \mb{0}}$ when ${h(\mb{x})=0}$. The function $h$ is a \emph{control barrier function} (CBF) for \eqref{eq:openloop1} on $\C$ if there exists ${\alpha\in\K_{\infty,\rm e}}$ such that for all ${\mb{x}\in\C}$:
\begin{equation}
\label{eq:cbf}
     \sup_{u\in\R^m} \dot{h}(\mb{x},\mb{u}) \triangleq \underbrace{\derp{h}{\mb{x}}(\mb{x})\mb{f}(\mb{x})}_{L_\mb{f}h(\mb{x})}+\underbrace{\derp{h}{\mb{x}}(\mb{x})\mb{g}(\mb{x})}_{L_\mb{g}h(\mb{x})}\mb{u}>-\alpha(h(\mb{x})).
\end{equation}
\end{definition}

\noindent Given a CBF $h$ for \eqref{eq:openloop1} and a corresponding ${\alpha\in\cal{K}_{\infty,\rm e}}$, we define the point-wise set of control values satisfying \eqref{eq:cbf} as:
\begin{equation}
    K_{\textrm{CBF}}(\mb{x}) \triangleq \left\{\mb{u}\in\R^m ~\left|~ \dot{h}(\mb{x},\mb{u})\geq-\alpha(h(\mb{x})) \right. \right\}.
\end{equation}

\begin{theorem}[\cite{ames2017control}]
Let ${\C\subset\R^n}$ be the 0-superlevel set of a continuously differentiable function ${h:\R^n\to\R}$ with ${\derp{h}{\mb{x}}(\mb{x}) \neq \mb{0}}$ when ${h(\mb{x})=0}$. If $h$ is a CBF for \eqref{eq:openloop1} on $\C$, then any Lipschitz continuous controller with ${\mb{k}(\mb{x}) \in K_{\rm CBF }(\mb{x})}$ for all ${\mb{x}\in\C}$ renders  \eqref{eq:closedloop1} safe with respect to the set $\C$.
\end{theorem}

\begin{example}
\label{ex:cbf}
Consider a dynamic system:
\begin{equation}
    \label{eq:ex_model}
    \dot{x}_1 =-x_2, \quad
    \dot{x}_2 =u,
\end{equation}
with state ${\mb{x} \in \R^2}$ and input ${u\in\R}$, a feedback controller:
\begin{equation}
\label{eq:ex_cont}
    k(\mb{x}) = x_1 - 2x_2 - 1,
\end{equation}
and the CBF candidate:
\begin{equation}
\label{eq:ex_h}
    h(\mb{x}) = x_1 - x_2,
\end{equation}
that defines the set $\C$ as: 
\begin{equation}
\label{eqn:ex_C}
    \C = \left\{ \mb{x} \in \R^2  ~\left|~  x_1 - x_2 \geq 0  \right. \right\}.
\end{equation}
The evolution of $h$ under \eqref{eq:closedloop1} is given by:
\begin{equation}
    \dot{h}(\mb{x}) = L_\mb{f}h(\mb{x}) + L_\mb{g}h(\mb{x}) k(\mb{x})
    = \underbrace{-x_1+x_2}_{-h(\mb{x})}+1 > -h(\mb{x}), \nonumber
\end{equation}
that is, choosing the extended class $\cal{K}_\infty$ function ${\alpha(r)=r}$ yields that ${k(\mb{x}) \in K_{\textrm{CBF}}(\mb{x})}$.
We present simulation results for the closed loop system in Fig.~\ref{fig:Example}(a), where all the trajectories initiated from different initial conditions ${\mb{x}(0) \in \C}$ safely approach the stable equilibrium point {$(1,0)$}.
\end{example}

\subsection{Input-to-State Safety}

Unmodeled effects and disturbances may make it infeasible for a state feedback controller ${\mb{k}(\mb{x})}$ to be implemented exactly. Instead, a potentially time-varying disturbance ${\mb{d}:\R_{\geq 0}\to\R^m}$ is added to the controller, such that ${\mb{u}=\mb{k}(\mb{x})+\mb{d}(t)}$, resulting in the closed loop system:
\begin{equation}
    \label{eq:closedloop2}
    \dot{\mb{x}} = \mb{f}(\mb{x})+\mb{g}(\mb{x})\mb{k}(\mb{x}) +\mb{g}(\mb{x})\mb{d}(t).
\end{equation}
The safety guarantees endowed by controllers satisfying ${\mb{k}(\mb{x}) \in K_{\textrm{CBF}}(\mb{x})}$ may no longer be valid for the disturbed closed loop system. 
Thus, we wish to design a safety-critical controller that ensures safety in the presence of disturbances. We consider the disturbed control system:
\begin{equation}
    \label{eq:openloop2}
    \dot{\mb{x}} = \mb{f}(\mb{x})+\mb{g}(\mb{x})\mb{u} +\mb{g}(\mb{x})\mb{d}(t),
\end{equation}
where the disturbance $\mb{d}$ is assumed to be bounded, that is, ${\Vert\mb{d}\Vert_{\infty} = \sup_{t \geq 0} \| \mb{d}(t) \| < \infty}$. With disturbances, we look for a larger set ${\mathcal{C}_{ {\delta}}\subset\R^n}$  parameterized by $\delta\geq 0$, i.e., ${\C \subseteq \C_{ {\delta}}}$,  that is forward invariant for all $\mb{d}$ satisfying ${\Vert\mb{d}\Vert_\infty \leq \delta}$. We require $\C_{ {\delta}}$ to grow monotonically with $ {\delta}$, and recover the original safe set in the absence of the disturbance, i.e., ${\C_{ {\delta}} \equiv \C}$ when ${ {\delta=0}}$. 
Thus, define a function ${h_{ {\delta}}:\R^n\times\R_{\geq 0}\to\R}$ as:
\begin{equation}
    \label{eq:hd_OISSf}
    h_{ {\delta}}(\mb{x}, {\delta}) \triangleq h(\mb{x}) + \gamma( {\delta}),
\end{equation}
with ${\gamma\in\cal{K}_{\infty}}$ and define $\C_{ {\delta}}$ as its 0-superlevel set:
\begin{align}
    \C_{ {\delta}} &\triangleq \left\{\mb{x} \in \R^n : h_{ {\delta}}(\mb{x}, {\delta}) \geq 0\right\}, \label{eq:Cd_OISSf1} \\
    \partial\C_{ {\delta}} &\triangleq \{\mb{x} \in \R^n : h_{ {\delta}}(\mb{x}, {\delta}) = 0\}, \label{eq:Cd_OISSf2}\\
    \textrm{Int}(\C_{ {\delta}}) &\triangleq \{\mb{x} \in \R^n : h_{ {\delta}}(\mb{x}, {\delta}) > 0\}. \label{eq:Cd_OISSf3}
\end{align}

\begin{definition}[\textit{Input-to-State Safety}]
Let ${\C\subset\R^n}$ be the 0-superlevel set of a continuously differentiable function ${h:\R^n\to\R}$. The system \eqref{eq:closedloop2} is \textit{input-to-state safe} (ISSf) if there exists ${\gamma \in \cal{K}_\infty}$ and $\delta\geq 0$ such that for all $\mb{d}$ satisfying ${\Vert\mb{d}\Vert_\infty\leq\delta}$, the set $\C_{ {\delta}}$ defined by \eqref{eq:Cd_OISSf1} is forward invariant. In this case, we refer to the original set $\C$ as an \emph{input-to-state safe set} (ISSf set).
\end{definition}

Given a controller ${\mb{k}(\mb{x})}$ that makes the undisturbed system \eqref{eq:closedloop1} safe with respect to the set $\C$ for a given CBF $h$, i.e., ${{\mb{k}(\mb{x}) \in K_{\textrm{CBF}}(\mb{x})}}$, we consider the following modification:
\begin{equation}
    \label{eq:OISSf_cont}
    \mb{u} = \mb{k}(\mb{x}) + \frac{1}{\epsilon_0}L_\mb{g}h(\mb{x})^\top,
\end{equation}
where ${\epsilon_0\in\R_{>0}}$ is a positive constant. Motivated by this controller, we give the definition of the \emph{input-to-state safe control barrier function}:
\begin{definition}[\textit{Input-to-State Safe Control Barrier Function (ISSf-CBF)}]
\label{def:def3}
Let ${\C\subset\R^n}$ be the 0-superlevel set of a continuously differentiable function ${h:\R^n\to\R}$ with ${\derp{h}{\mb{x}}(\mb{x}) \neq \mb{0}}$ when ${h(\mb{x})=0}$. Then $h$ is an \emph{input-to-state safe control barrier function} (ISSf-CBF) for \eqref{eq:openloop2} on $\C$ if there exists ${\alpha\in\K_{\infty,\rm e}}$ and ${\epsilon_0>0}$ such that for all ${\mb{x}\in\R^n}$:
\begin{equation}
\label{eq:oissf-cbf} 
\sup_{\mb{u}\in\R^m} \left[ L_\mb{f}h(\mb{x}) + L_\mb{g}h(\mb{x})\mb{u} \right] >
-\alpha(h(\mb{x})) + \frac{\Vert L_{\mb{g}}h(\mb{x}) \Vert^2}{\epsilon_0}.
\end{equation}
\end{definition}
\noindent As with CBFs, we may define the point-wise set of control values satisfying \eqref{eq:oissf-cbf}:
\begin{align}
\resizebox{1\hsize}{!}{
    ${K_{\textrm{ISSf}}(\mb{x}) \triangleq \left\{\mb{u}\in\R^m ~\left|~  \dot{h}(\mb{x},\mb{u})\geq-\alpha(h(\mb{x})) + \frac{\Vert L_{\mb{g}}h(\mb{x}) \Vert^2}{\epsilon_0}  \right. \right\}.}$
}     
\end{align}

\begin{theorem}[{\cite{Kolathaya:2019ISSf}}]    \label{theo:ISSf}
Let ${\C\subset\R^n}$ be the 0-superlevel set of a continuously differentiable function ${h:\R^n\to\R}$ with ${\derp{h}{\mb{x}}(\mb{x}) \neq \mb{0}}$ when ${h(\mb{x})=0}$ and $\delta\geq 0$. If $h$ is an ISSf-CBF for \eqref{eq:openloop2} on $\C$, then for any Lipschitz continuous controller with ${\mb{k}(\mb{x})\in K_{\rm ISSf}}$ for all ${\mb{x}\in\R^n}$ and for all $\mb{d}$ satisfying ${\Vert\mb{d}\Vert_\infty \leq \delta}$, the system \eqref{eq:closedloop2} is safe with respect to $\C_{ {\delta}}$ defined as in \eqref{eq:Cd_OISSf1} with  ${\gamma\in\K_\infty}$ defined as:
\begin{equation}    \label{eq:OISSf_gamma}
    \gamma( {\delta}) \triangleq -\alpha^{-1}\left(-\frac{\epsilon_0  {\delta}^2}{4}\right),
\end{equation}
where ${\alpha^{-1} \in \K_{\infty,{\rm e}}}$. This implies $\C$ is an ISSf set.
\end{theorem}

\begin{remark}
The original ISSf-CBF definition proposed in \cite{Kolathaya:2019ISSf} requires the condition:
\begin{equation}
\label{eqn:issfcbfold}
\sup_{\mb{u}\in\R^m} \left[ L_\mb{f}h(\mb{x}) + L_\mb{g}h(\mb{x})(\mb{u+d}) \right] >
-\alpha(h(\mb{x})) - \iota ( \Vert \mb{d} \Vert_\infty ),
\end{equation}
for some ${\iota \in \K_{\infty}}$. It also proves that a function satisfying \eqref{eq:oissf-cbf} meets the condition in \eqref{eqn:issfcbfold} for $\iota$ defined as:
\begin{equation}
    \iota ( \Vert \mb{d} \Vert_\infty ) \triangleq \frac{\epsilon_0 \Vert \mb{d} \Vert_\infty^2}{4}.
\end{equation} 
We use the more specific definition in \eqref{eq:oissf-cbf} as it is better suited for the controller design presented in this letter. 
\end{remark}

As ${\alpha^{-1}\in\cal{K}_{\infty,{\rm e}}}$, a smaller $\epsilon_0$ implies a smaller value of ${\gamma( {\delta})}$ for a given $ {\delta\geq 0}$, which reduces the difference between the sets $\C$ and $\C_{ {\delta}}$. However, taking $\epsilon_0$ to be small increases the right hand side of \eqref{eq:oissf-cbf}, and forces a more restrictive safety condition to be met by $\mb{k}$. Controllers satisfying this more restrictive condition may lead to undesirable performance as illustrated by the example below.

\begin{figure}[t]
    \centering
    \includegraphics[trim=70 70 65 70,clip, scale = 0.45]{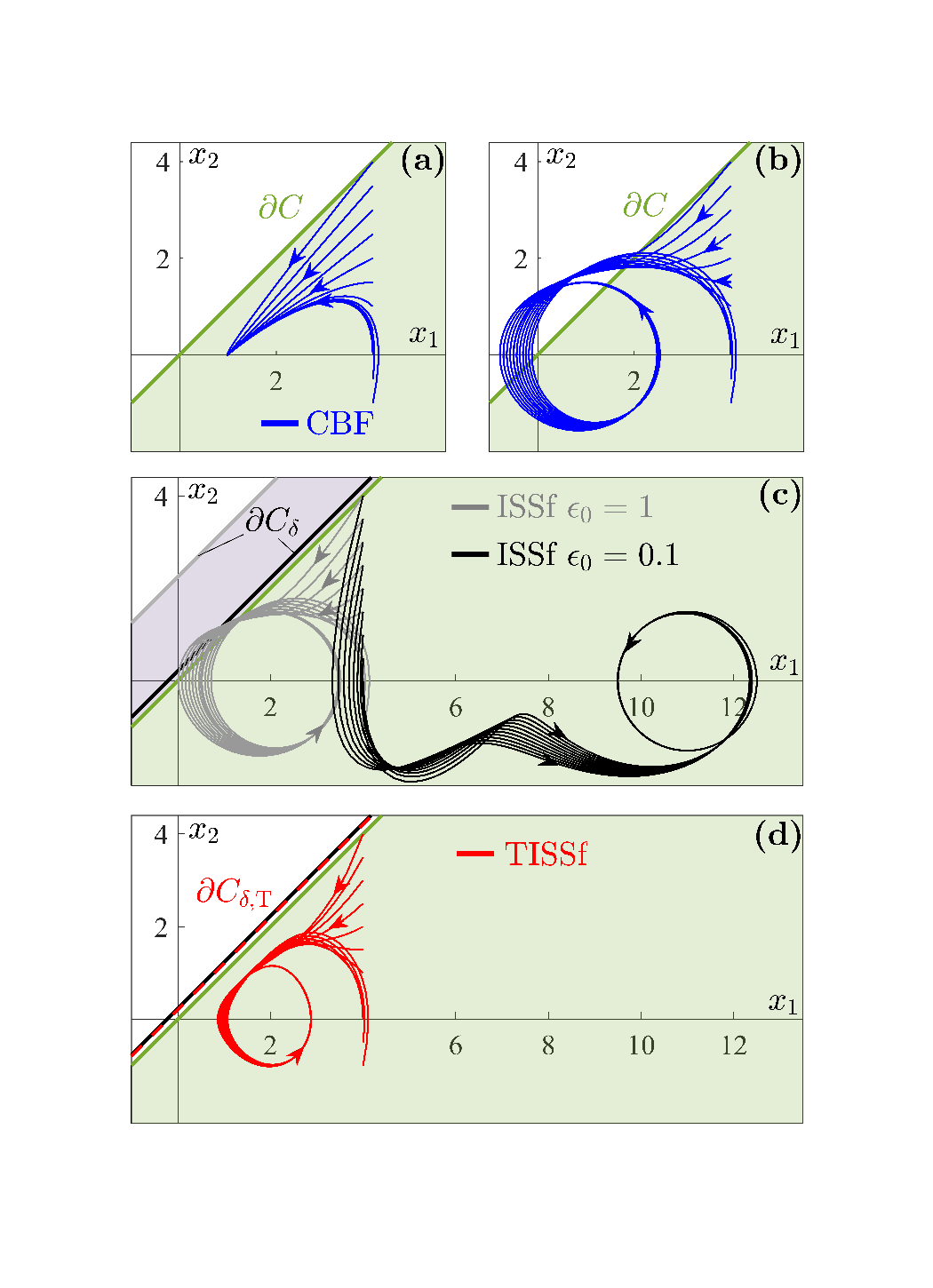}
    \caption{The sets $\C$, $\C_{ {\delta}}$ and $\C_{ {\delta},\rm T}$ (shaded) and simulation results for Examples \ref{ex:cbf}-\ref{ex:tissf}. (a) The boundary $\partial \C$ (green) and simulated trajectories with controller \eqref{eq:ex_cont} without disturbance. (b) Trajectories with disturbance. (c) The boundary $\partial \C_{ {\delta}}$ for ${\epsilon_0=0.1}$ (gray) and ${\epsilon_0=1}$ (black) and simulation results for controller \eqref{eq:ex_contISSf}. (d) The boundary $\partial \C_{ {\delta},\rm T}$ (red) and simulation results for  controller \eqref{eq:ex_contTISSf}.  }
    \label{fig:Example}
    \vspace{-3mm}
\end{figure}

\begin{example}
\label{ex:dist}
We now introduce a disturbance to the example:
\begin{equation}
    \label{eq:ex_modeld}
    \dot{x}_1 =-x_2, \quad
    \dot{x}_2 =u+d(t),
\end{equation}
where ${d:\R_{\geq 0}\to\R}$. Fig.~\ref{fig:Example}(b) depicts the simulation results with the controller $k(\mb{x})$ defined in \eqref{eq:ex_cont} for the harmonic disturbance ${d(t)= {\delta}\sin{t}}$ with ${ {\delta}=3}$. We see that the disturbance makes the state trajectories leave $\C$ periodically. 
According to \eqref{eq:OISSf_cont}, we consider the modified controller:
\begin{equation} \label{eq:ex_contISSf}
     u = k(\mb{x}) + \frac{L_\mb{g}h(\mb{x})}{\epsilon_0} = x_1 - 2x_2 - 1 - \frac{1}{\epsilon_0},
\end{equation}
cf.~\eqref{eq:ex_cont}. The evolution of $h$ under \eqref{eq:closedloop2} is given by:
\begin{equation}
    L_\mb{f}h(\mb{x}) + L_\mb{g}h(\mb{x})k(\mb{x}) = \underbrace{-x_1 + x_2}_{-h(\mb{x})} + 1 + \frac{1}{\epsilon_0} > -h(\mb{x}) + \frac{1}{\epsilon_0}, \nonumber
\end{equation}
such that with ${\alpha(r)=r}$, $h$ is an ISSf-CBF for \eqref{eq:ex_modeld} on the set $\C$ defined in \eqref{eqn:ex_C}. Furthermore, with ${-\alpha^{-1}(-r)=r}$, 
we have ${\gamma( {\delta})=\frac{\epsilon_0  {\delta}^2}{4}}$, yielding:
\begin{equation}
    \C_{ {\delta}} = \left\{ \mb{x} \in \R^2  ~\left|~  x_1 - x_2 + \frac{\epsilon_0  {\delta}^2}{4} \geq 0  \right. \right\}.
\end{equation}
Figure~\ref{fig:Example}(c) portrays the boundary $\partial \C_{ {\delta}}$ for ${\epsilon_0=1}$ (gray) and ${\epsilon_0=0.1}$ (black). A larger $\epsilon_0$ implies a larger gap between the original set $\C$ and the forward invariant set $\C_{ {\delta}}$, and as a result, gives way to the trajectories leaving $\C$. In contrast, a smaller $\epsilon_0$ shifts $\C_{ {\delta}}$ closer to $\C$,  yielding trajectories that stay in $\C$. This, however, comes with an expense of substantially effecting the performance as the trajectories are pushed further inside $\C$.
\end{example}

\section{Main Result}\label{sec:mainresult}

In this section, we present the main result of the paper by introducing a new method for characterizing safety in the presence of disturbances. It uses a more general definition of the set $\C_{ {\delta}}$ to enable synthesis of controllers that can ensure safety without compromising performance. 

The previous specification of $h_{ {\delta}}$ and $\gamma$ as in \eqref{eq:hd_OISSf} and \eqref{eq:OISSf_gamma}, respectively, implies that the difference ${h_{ {\delta}}(\mb{x})-h(\mb{x})}$ is constant
for all $\mb{x}\in\C_{ {\delta}}$ for a given $ {\delta}$. 
In other words, requiring a constant $\epsilon_0$ imposes strong restrictions on the structure of $h_{ {\delta}}(\mb{x})$ and $\C_{ {\delta}}$. As a result, prioritizing safety with a smaller $\epsilon_0$ may lead to overcompensation and may affect the performance in an undesirable fashion. We wish to find a new set that is still forward invariant, but allows more flexibility in designing controllers. To this end, define the function $h_{ {\delta},T}:\R^n\times\R_{\geq 0}\to\R$ as:
\begin{equation}
    \label{eq:tissf_hd}
    h_{ {\delta},\rm T}(\mb{x}, {\delta}) = h(\mb{x}) + \gamma_{\rm T} (h(\mb{x}), {\delta}),
\end{equation}
with ${\gamma_{\rm T}:\R\times\R_{\geq0}\to\R_{\geq0}}$ continuously differentiable in its first argument and ${\gamma_{\rm T}(a,\cdot) \in \cal{K}_{\infty}}$ for all ${a\in\R}$. Indeed, $h_{ {\delta}}$
defined by \eqref{eq:hd_OISSf} is a special case of $h_{ {\delta},\rm T}$ defined by \eqref{eq:tissf_hd}. We define $\C_{ {\delta},\rm T}$ as the 0-superlevel set of the function $h_{ {\delta},\rm T}$:
\begin{align}
    \C_{ {\delta},\rm T} &\triangleq \left\{\mb{x} \in \R^n : h_{ {\delta},\rm T}(\mb{x}, {\delta}) \geq 0\right\}, \label{eq:CdT1} \\
    \partial\C_{ {\delta},\rm T} &\triangleq \{\mb{x} \in \R^n : h_{ {\delta},\rm T}(\mb{x}, {\delta}) = 0\}, \label{eq:CdT2}\\
    \textrm{Int}(\C_{ {\delta},\rm T}) &\triangleq \{\mb{x} \in \R^n : h_{ {\delta},\rm T}(\mb{x}, {\delta}) > 0\}. \label{eq:CdT3}
\end{align}
Note that ${\C \subset \C_{ {\delta},\rm T}}$ for $ {\delta}>0$. In the absence of disturbances (${ {\delta}=0}$) we recover the original set ${(\C_{ {\delta},\rm T} \equiv \C)}$ as ${h_{ {\delta},\rm T}(\mb{x},0)=h(\mb{x})}$. Also, $\C_{ {\delta},\rm T}$ grows monotonically with $ {\delta}$.
Analogous to \eqref{eq:OISSf_cont}, we propose the controller:
\begin{equation}\label{eq:TISSf_cont}
    \mb{u} = \mb{k}(\mb{x}) + \frac{1}{\epsilon(h(\mb{x}))}L_\mb{g}h(\mb{x})^\top,
\end{equation}
where ${\epsilon:\R\to\R_{>0}}$ is a continuously differentiable function and ${\mb{k}(\mb{x}) \in K_{\textrm{CBF}}(\mb{x})}$. This controller motivates a generalization of Definition~\ref{def:def3}, and a corresponding safety result.
\begin{definition}[\textit{Tunable Input-to-State Safe Control Barrier Function (TISSf-CBF)}]
Let ${\C\subset\R^n}$ be the 0-superlevel set of a continuously differentiable function ${h:\R^n\to\R}$ with ${\derp{h}{\mb{x}}(\mb{x}) \neq \mb{0}}$ when ${h(\mb{x})=0}$. Then $h$ is a \emph{tunable input-to-state safe control barrier function} (TISSf-CBF) for \eqref{eq:openloop2} on $\C$ with continuously differentiable function ${\epsilon:\R\to\R_{>0}}$ if there exists ${\alpha \in \K_{\infty,\rm e}}$ such that for all ${\mb{x}\in \R^n}$:
\begin{equation}
\label{eq:tissf-cbf} 
\sup_{\mb{u}\in\R^m}  \left[ L_\mb{f}h(\mb{x}) + L_\mb{g}h(\mb{x})\mb{u} \right] >
-\alpha(h(\mb{x})) + \frac{\Vert L_{\mb{g}}h(\mb{x}) \Vert^2}{\epsilon(h(\mb{x}))}.
\end{equation}
\end{definition}
\noindent As with ISSf-CBFs, we may define the point-wise set of control values satisfying \eqref{eq:tissf-cbf}:
\begin{align}
\resizebox{1\hsize}{!}{
    ${K_{\textrm{TISSf}}(\mb{x}) \triangleq \left\{\mb{u}\in\R^m ~\left|~  \dot{h}(\mb{x},\mb{u})\geq-\alpha(h(\mb{x})) + \frac{\Vert L_{\mb{g}}h(\mb{x}) \Vert^2}{\epsilon(h(\mb{x}))} \right. \right\}.}$
}     
\end{align}

\begin{theorem}\label{theo:TISSf}
Let ${\C\subset\R^n}$ be the 0-superlevel set of a continuously differentiable function ${h:\R^n\to\R}$ and $\delta \geq 0$. If $h$ is a TISSf-CBF for \eqref{eq:openloop2} on $\C$ with continuously differentiable function ${\epsilon:\R\to\R_{>0}}$ such that ${\alpha^{-1}\in \K_{\infty,\rm e}}$ is continuously differentiable and $\epsilon$ satisfies:
\begin{equation}
    \label{eq:eps_geq0}
    \der{\epsilon}{r}(h(\mb{x})) \geq 0,
\end{equation}
then for any Lipschitz continuous controller with ${\mb{k}(\mb{x})\in K_{\rm TISSf}(\mb{x})}$ for all ${\mb{x}\in\R^n}$ and for all $\mb{d}$ satisfying ${\Vert\mb{d}\Vert_\infty \leq \delta}$, the system \eqref{eq:closedloop2} is safe with respect to $\C_{ {\delta},T}$ defined as in \eqref{eq:CdT1}-\eqref{eq:CdT3} with ${\gamma_{\rm T}:\R\times\R_{\geq0}}$ defined as:
\begin{equation}    \label{eq:tissf_gamma}
    \gamma_{\rm T}(h(\mb{x}), {\delta}) \triangleq -\alpha^{-1}\left(-\frac{\epsilon(h(\mb{x}))  {\delta}^2}{4}\right).
\end{equation}
\end{theorem}

\begin{proof}
Our goal is to show that the set $\C_{ {\delta},\rm T}$ defined by \eqref{eq:CdT1}-\eqref{eq:CdT3} is forward invariant. For a controller satisfying ${\mb{k}(\mb{x})\in K_{\textrm{TISSf}}(\mb{x})}$ for all ${\mb{x}\in\R^n}$, we have:
\begin{equation}    \label{eq:proof1}
\begin{split}
    \dot{h}(\mb{x},t) =& L_{\mb{f}}h(\mb{x}) + L_{\mb{g}}h(\mb{x}) (\mb{k}(\mb{x}) + \mb{d}(t)) \\
    \geq& -\alpha(h(\mb{x})) + \frac{\Vert L_{\mb{g}}h(\mb{x}) \Vert^2}{\epsilon(h(\mb{x}))}  + L_{\mb{g}}h(\mb{x}) \mb{d}(t).
\end{split}
\end{equation}
Noting that:
\begin{equation*}
L_{\mb{g}}h(\mb{x}) \mb{d}(t) \geq -\| L_{\mb{g}}h(\mb{x}) \| \|\mb{d}\|_\infty  {\geq -\| L_{\mb{g}}h(\mb{x}) \|\delta}
\end{equation*}
and ${\epsilon(h(\mb{x}))>0}$ for all ${\mb{x}\in\R^n}$ and ${t\geq 0}$, adding and subtracting ${\frac{\epsilon(h(\mb{x}))  {\delta}^2}{4}}$, and completing the squares yields:
\begin{equation}   \label{eq:proof3}
    \dot{h}(\mb{x},t) \geq -\alpha(h(\mb{x})) - \frac{\epsilon(h(\mb{x}))  {\delta}^2}{4}.
\end{equation}
Next, taking the time derivative of the function ${h_{ {\delta},\rm T}}$ defined by \eqref{eq:tissf_hd} yields:
\begin{align}    
    \dot{h}_{ {\delta},\rm T}(\mb{x}, {\delta},t) =
    \left(1 + \derp{\gamma_{\rm T}}{h}(h(\mb{x}), {\delta}) \right) \dot{h}(\mb{x},t).
    \label{eq:proof4}
\end{align}
As $\epsilon$ satisfies \eqref{eq:eps_geq0} and $\gamma_{\rm T}$ is defined as in \eqref{eq:tissf_gamma}, we have:
\begin{equation}
    \label{eq:Apos}
   1 + \derp{\gamma_{\rm T}}{h}(h(\mb{x}), {\delta})  > 0.
\end{equation}
Substituting \eqref{eq:proof3} into \eqref{eq:proof4}, we obtain:
\begin{equation}
    \label{eq:proof6}
    \begin{split}
    & \dot{h}_{ {\delta},\rm T}(\mb{x}, {\delta},t) \geq \\
    & \left(1 + \derp{\gamma_{\rm T}}{h}(h(\mb{x}), {\delta}) \right) \left(  -\alpha(h(\mb{x})) - \frac{\epsilon(h(\mb{x}))  {\delta}^2}{4}  \right). \nonumber
    \end{split}
\end{equation}
Next, we consider a state ${\mb{x}\in\partial \C_{ {\delta},\rm T}}$, such that ${h_{ {\delta},\rm T}(\mb{x})=0}$, for which \eqref{eq:tissf_hd} and \eqref{eq:tissf_gamma} imply: 
\begin{equation}
    -\alpha(h(\mb{x})) - \frac{\epsilon(h(\mb{x}))  {\delta}^2}{4} = 0,
\end{equation}
yielding:
\begin{equation}
    \label{eq:proof7}
    \dot{h}_{ {\delta},\rm T}(\mb{x}, {\delta},t) \geq 0.
\end{equation}
Additionally, we have ${-\alpha(h(\mb{x})) \geq 0}$ when ${h_{ {\delta},\rm T}(\mb{x})=0}$ by construction. Thus, the strict inequality in \eqref{eq:tissf-cbf} requires that ${\derp{h}{\mb{x}}(\mb{x}) \neq \mb{0}}$ for ${\mb{x}\in\partial\C_{\delta,T}}$.
Finally, we have: 
\begin{equation}
    \derp{h_{ {\delta},\rm T}}{\mb{x}}(\mb{x}, {\delta}) = \underbrace{\left(1 + \derp{\gamma_{\rm T}}{h}(h(\mb{x}), {\delta}) \right)}_{>0} \derp{h}{\mb{x}}(\mb{x}) \neq 0,
\end{equation}
using \eqref{eq:Apos}.
Therefore, Nagumo's theorem \cite{Nagumo:1942} implies the set $\C_{ {\delta},\rm T}$ is forward invariant as ${h_{ {\delta},\rm T}(\mb{x}, {\delta})=0}$ implies ${\dot{h}_{ {\delta},\rm T}(\mb{x}, {\delta},t) \geq 0}$, and ${\derp{h_{ {\delta},\rm T}}{\mb{x}}(\mb{x}, {\delta}) \neq 0}$.
\end{proof}
\begin{remark}
We note that the condition on $\epsilon$ in \eqref{eq:eps_geq0} is stronger than necessary, but is an easily verifiable design condition. In particular, $\epsilon$ only needs to satisfy that for $\delta>0$ and ${\mb{x}\in\R^n}$:
\begin{equation}    \label{eq:eps_cond}
    \der{\epsilon}{r}(h(\mb{x})) > -\frac{4}{ {\delta}^2} 
    \dfrac{1}{D\alpha^{-1}(-\epsilon(h(\mb{x}))\delta^2/4)}.
\end{equation}
Noting that ${\alpha^{-1}\in\K_{\infty,{\rm e}}}$ and is continuously differentiable implies ${0\leq D\alpha^{-1}(-\epsilon(h(\mb{x}))\delta^2/4)< \infty}$ for all ${\mb{x}\in\R^n}$. The right-hand side of \eqref{eq:eps_cond} approaches $-\infty$ as ${\delta \to 0}$ or ${D\alpha^{-1}(-\epsilon(h(\mb{x}))\delta^2/4)\to 0}$, making $\epsilon$ unconstrained. 
\end{remark}

\begin{example}
\label{ex:tissf}
For the disturbed system in Example \ref{ex:dist} with ${\delta = 3}$, we pick the following differentiable function:
\begin{equation}
    \label{eq:ex_eps}
    \epsilon(h(\mb{x})) \triangleq \epsilon_0 {\rm e}^{ {\lambda} h(\mb{x})},
\end{equation}
with constants ${\epsilon_0, {\lambda} \in \R_{>0}}$.
Considering the controller:
\begin{equation}  \label{eq:ex_contTISSf}
    u = k(\mb{x}) + \frac{L_\mb{g}h(\mb{x})}{\epsilon(h(\mb{x}))} 
    = x_1 - 2x_2 - 1 - \frac{1}{\epsilon_0 {\rm e}^{ {\lambda} (x_1-x_2)}},
\end{equation}
it can be shown that $h$ as defined in \eqref{eq:ex_h} is a TISSf-CBF for ${\alpha(r)=r}$. Furthermore, the choice of the function $\epsilon$ with ${\epsilon_0,\lambda > 0}$ in \eqref{eq:ex_eps} satisfies the condition \eqref{eq:eps_geq0}. Thus, the set:
\begin{align}
	\resizebox{.88\hsize}{!}{
		$\C_{ {\delta},\rm T} = \left\{ \mb{x} \in \R^2  ~\left|~  x_1 - x_2 + \frac{\epsilon_0 {\rm e}^{\lambda (x_1-x_2)}  {\delta}^2}{4} \geq 0  \right. \right\},$
	}     
\end{align}
is forward invariant. It is noted that ${\lambda=0}$ recovers ISSf setup with ${\epsilon(h(\mb{x})) \equiv \epsilon_0}$, whereas a larger $\lambda$ pulls $\C_{ {\delta},\rm T}$ closer to the safe set $\C$, and decreases the effect of the corresponding term in the controller \eqref{eq:ex_contTISSf} for ${h(\mb{x})>0}$.
We depict the set $\C_{ {\delta},\rm T}$ in Fig.~\ref{fig:Example}(d) considering ${\epsilon_0=\rm{e}^{-2}}$, ${\lambda=2}$ and along with simulation results. All solution trajectories stay within the set $\C_{ {\delta},\rm T}$ that is close to $\C$, and the overcompensation inside the set $\C$ is prevented as $\epsilon(h(\mb{x}))$ takes larger values when ${h(\mb{x}) \gg 0}$.
\end{example}

\begin{remark}
Ensuring the forward invariance of a slightly larger set suggests modifying the set $\C$ in the presence of a disturbance. Specifically, considering a set $\overline{\C} \subseteq \C$ such that $\overline{\C}_{\delta,\rm T} \subseteq \C$ implies the safety of the original set $\C$.
\end{remark}

\begin{remark}
Rather than utilizing \eqref{eq:TISSf_cont} by modifying an existing safe controller ${\mb{k}(\mb{x})\in K_{\textrm{CBF}}(\mb{x})}$, the condition \eqref{eq:tissf-cbf} can be utilized to synthesize an optimization-based controller via the following quadratic program:
\begin{align}
\tag{TISSf-QP}
\label{eq:TISSf-QP}
\mb{k}_{\rm QP}(\mb{x}) = & \,\,\underset{\mb{u} \in \R^m}{\argmin}    \quad \frac{1}{2} \Vert \mb{u}-\mb{k}(\mb{x})  \Vert^2  \\
 \mathrm{s.t.} \quad & L_\mb{f}h(\mb{x}) + L_\mb{g}h(\mb{x})\mb{u} >
-\alpha(h(\mb{x})) + \frac{\Vert L_{\mb{g}}h(\mb{x}) \Vert^2}{\epsilon(h(\mb{x}))}, \nonumber
\end{align}
that may intervene less compared to \eqref{eq:TISSf_cont}.

\end{remark}

\section{Input-to-state Safety for Automated Trucks}\label{sec:truck}

Here we implement previously introduced \textit{tunable input-to-state safe control barrier functions} (TISSf-CBF) to design the longitudinal controller of a connected automated truck while responding to the motion of a connected vehicle ahead. We use a simplified model to design the controller and we demonstrate that it can maintain safety in real-world safety-critical scenario by simulating a high-fidelity vehicle model. 

Consider the simplified model for the system:
\begin{equation}   \label{eq:truck_model}
        \dot{D} =v_{\rm L}-v, ~~~
        \dot{v} =u+d(t), ~~~
        \dot{v}_{\rm L} = a_{\rm L},
\end{equation}
where $D$ denotes the bumper-to-bumper headway distance between the truck and the vehicle ahead, $v$ is the longitudinal velocity of the truck, while $v_{\rm L}$ and $a_{\rm L}$ are longitudinal velocity and acceleration of the preceding vehicle. The state is defined by ${\mb{x} = [D, v, v_{\rm L}] \in \R^3}$ while $u$ denotes the input. The input disturbance $d(t)$ represents the unmodeled dynamics, i.e., rolling resistance, air drag, powertrain dynamics and delays related to sensing, computation and communication. We remark that while the distance $D$ and the velocities ${v, v_{\rm L}}$ can be measured by sensors, to obtain the acceleration signal
$a_{\rm L}$ V2X communication is needed \cite{Chaozhe:18}. That is why we refer to the controller below as connected cruise control rather than adaptive cruise control. Finally, to incorporate physical limitations we prescribe bounds for the input and the states:
\begin{align}
        \label{eq:truck_bounds}
    u \in [-\underline{a}, \overline{a}], ~~~
    a_{\rm L} \in [-\underline{a}_{\rm L}, \overline{a}_{\rm L}], ~~~
    v,v_{\rm L} \in [0,\overline{v}], 
\end{align}
where ${\underline{a}=6~[{\rm m}/{\rm s^2}]}$, ${\overline{a}=2~[{\rm m}/{\rm s^2}]}$, ${\underline{a}_{\rm L}=10~[{\rm m}/{\rm s^2}]}$, ${\overline{a}_{\rm L}=3~[{\rm m}/{\rm s^2}]}$ and ${\overline{v}=20~[{\rm m}/{\rm s}]}$ are considered.

In order to ensure safety the truck needs to keep a safe distance from the preceding vehicle which may depend on the velocities. This leads to the safety function candidate:
\begin{equation}  \label{eq:truck_h}
    h(\mb{x}) = D - \hat{h}(v,v_{\rm L}),
\end{equation}
where we use
\begin{equation} \label{eq:truck_CDF}
    \hat{h}(v,v_{\rm L}) = D_{\rm sf} + \theta v + \eta v_{\rm L} + \xi v^2 + \zeta v v_{\rm L} + \omega v_{\rm L}^2.
\end{equation}
The parameters ${D_{\rm sf}=2~[{\rm m}]}$, ${\theta=1.1~[{\rm s}]}$, ${\eta=0.6~[{\rm s}]}$, and ${\xi=-\zeta=-\omega=0.03~[{\rm s^2}/{\rm m}]}$
are chosen such that the truck is kept beyond a critical time headway of 1 second while considering the physical bounds \eqref{eq:truck_bounds}.
It can be shown that for \eqref{eq:truck_h}-\eqref{eq:truck_CDF} we have ${\derp{h}{\mb{x}} \neq 0}$ when ${h(\mb{x})=0}$.

We define the set: 
\begin{equation}
    \C = \left\{\mb{x} \in \R^3 ~\left|~ D - \hat{h}(v,v_{\rm L}) \geq 0 \right. \right\},
\end{equation}
and to render it safe, we utilize a feedback controller 
\begin{equation} \label{eq:truck_nominal}
k(\mb{x})= k_1 (V(D)-v) + k_2 (v_{\rm L}-v),
\end{equation}
where ${k_1,k_2 \in \R}$ are the controller parameters. The first term in \eqref{eq:truck_nominal} contains the range policy function ${V:\R\to\R_{\geq 0}}$:
\begin{equation} 	\label{eq:truck_V}
V(D) = \max\big\{0,\min\{\kappa(D-D_{\rm st}), \overline{v} \}\big\},
\end{equation}
where $D_{\rm st}$ is the desired stopping distance and ${1/\kappa}$ defines the desired time headway. The second term in \eqref{eq:truck_nominal} responds to the speed mismatch. Considering ${\alpha(r)=r}$ one may show that the parameters ${k_1=0.7~[1/\text{s}]}$, ${k_2=0.75~[1/\text{s}]}$, ${\kappa=0.7~[1/\text{s}]}$, ${D_{\rm st}=7~[\text{m}]} $ yield ${k(\mb{x}) \in K_{\rm CBF}}$; see \cite{Chaozhe:18}.

In order to incorporate real-world disturbances, numerical simulations are carried out using a high fidelity truck model built in TruckSim and Simulink. This model contains details about the engine, clutch, gearbox, tires and mechanical/hydraulic braking components which inevitably delay the realization of the longitudinal acceleration command and considered as disturbance in the simple model \eqref{eq:truck_model}. Pre-recorded experimental data is used to represent the preceding vehicle's speed $v_{\rm L}$ and acceleration $a_{\rm L}$; see Fig.~\ref{Fig:SimResults}(b,d). In particular, the recorded data correspond to an emergency braking scenario in city traffic where the preceding vehicle decelerates from ${15~{\rm [m/s]}}$ to a full stop with acceleration reaching ${-8~{\rm [m/s^2]}}$. The simulation results are presented in Fig.~\ref{Fig:SimResults} as blue curves. While the truck avoids the crash, it is unable to maintain safety ($h$ becomes negative in panel (c)) as the controller \eqref{eq:truck_nominal} is designed using the model \eqref{eq:truck_model} with no disturbance.

\begin{figure}[t]
	\centering
	\includegraphics[trim=15 0 30 20,clip, width=0.8\linewidth]{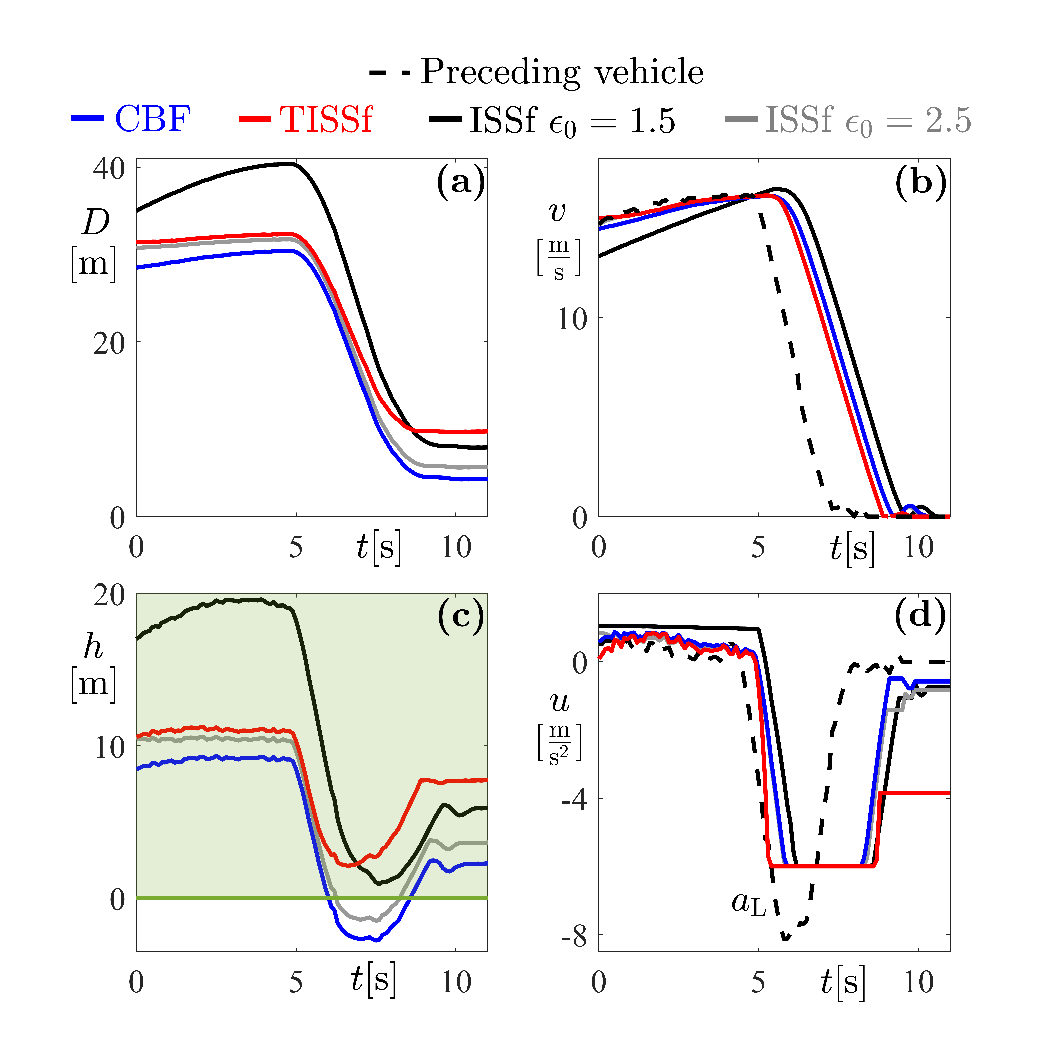}
	\caption{High-fidelity simulation results showing (a) distance, (b) velocities, (c) the barrier function $h$ defined by \eqref{eq:truck_h}, and (d) input $u$. Simulations are carried out with the CBF controller \eqref{eq:truck_nominal} (blue), the ISSf-CBF controller \eqref{eq:truck_contISSf} for ${\epsilon_0=1.5}$ (black) and ${\epsilon_0=2.5}$ (gray), and the TISSf-CBF controller \eqref{eq:truck_contTISSf} (red). }  \label{Fig:SimResults}
	   \vspace{-4mm}
\end{figure}

Next we modify the controller \eqref{eq:truck_nominal} as:
\begin{equation} 
\label{eq:truck_contISSf}
k_{\rm ISSf}(\mb{x}) = k_1 (V(D)-v) + k_2 (v_{\rm L}-v) - \frac{1}{\epsilon_0}\derp{\hat{h}}{v}(v,v_{\rm L}),
\end{equation}
(cf.~\eqref{eq:OISSf_cont}) where we used  ${L_\mb{g}h(\mb{x})=-\derp{\hat{h}}{v}(v,v_{\rm L})}$. Since $h$ is an ISSf-CBF for any ${\epsilon_0 > 0}$ the set:
\begin{equation}
    \C_ {\delta} = \left\{\mb{x} \in \R^3 ~\left|~ D - \hat{h}(v,v_{\rm L}) { + \frac{\epsilon_0  {\delta}^2}{4} \geq 0} \right. \right\},
\end{equation}
is forward invariant according to Theorem~\ref{theo:ISSf}. The corresponding simulations are shown in Fig.~\ref{Fig:SimResults} by black and gray curves for two different values of $\epsilon_0$. Panel (c) shows that the system leaves the original set $\C$ for ${\epsilon_0=2.5}$ (gray) as indicated by ${h<0}$. Choosing ${\epsilon_0=1.5}$ (black) ensures that ${h>0}$, it substantially affects the performance by making the truck to keep larger distances even when traveling with a constant speed (which would likely invite other vehicles to cut in). 

Finally, we consider the TISSf-CBF setting and modify the controller \eqref{eq:truck_nominal} as:
\begin{equation} 
\label{eq:truck_contTISSf}
\begin{split}
k_{\rm TISSf}(\mb{x})&= k_1 (V(D)-v) + k_2 (v_{\rm L}-v) 
\\
& - \frac{1}{\epsilon_0{\rm e}^{  {\lambda} (D-\hat{h}(v,v_{\rm L})) }}\derp{\hat{h}}{v}(v,v_{\rm L}),
\end{split}
\end{equation}
with ${\epsilon(h(\mb{x})}$ as defined in \eqref{eq:ex_eps}; cf.~\eqref{eq:TISSf_cont}. It can be verified that any parameter combination  ${\epsilon_0, \lambda > 0}$ make $h$ a TISSf-CBF. Thus, according to Theorem~\ref{theo:TISSf}, the set: 
\begin{align}
\resizebox{.88\hsize}{!}{
    ${\C_{ {\delta},\rm T} = \bigg\{ \mb{x} \in \R^3 \bigg|
     D -\hat{h}(v,v_{\rm L}) { + \frac{\epsilon_0 {\rm e}^{  \lambda (D-\hat{h}(v,v_{\rm L})) }  {\delta}^2}{4} } \geq 0  \bigg\},}$
}     
\end{align}
is forward invariant. The corresponding simulation results are shown in Fig.~\ref{Fig:SimResults} as red curves for parameters  ${\epsilon_0={\rm e^{-5}}~[{\rm m}]}$ and ${\lambda=0.5~[1/{\rm m}]}$. Observe that the system stays within the original set $\C$ without leaving a large distance headway at a steady state speed.

\section{Conclusion} \label{sec:conclusion}

In this letter, we first reviewed the notion of \textit{input-to-state safety} formulated by \textit{input-to-state safe control barrier functions} (ISSf-CBF), and provided the conditions for the forward invariance of a set under input disturbance. We then presented the new \textit{tunable input-to-state safe control barrier functions} (TISSf-CBF) to remedy the lack of tunability of the previous setup. We demonstrated the effectiveness of the new method in simulation environment with a high fidelity automated truck model.
Future work will include implementing a safety-critical control based on TISSf-CBF to a real automated truck and ensuring safety experimentally.



\end{document}